\documentclass[pra,twocolumn,showpacs,amsmath,amssymb,a4paper]{revtex4}
\usepackage{graphicx}
\usepackage{mathrsfs}
\usepackage{amssymb}
\usepackage{amsmath}
\usepackage{amsthm}
\usepackage{bm}
\usepackage{hyperref}

\newtheorem{proposition}{Proposition}

\newtheorem{corollary}{Corollary}

\theoremstyle{definition}

\newtheorem{example}{Example}

\newcommand{\bra}[1]{\langle #1|}
\newcommand{\ket}[1]{| #1 \rangle }
\newcommand{\ip}[2]{{\langle #1|}{ #2 \rangle }}
\newcommand{\tr}[1]{{\rm tr}[#1]}
\newcommand{\be}{\begin{eqnarray}}
\newcommand{\ee}{\end{eqnarray}}

\begin{document}

\title{PPT-inducing, distillation-prohibiting, and entanglement-binding quantum channels}

\author{Sergey N. Filippov}

\affiliation{Moscow Institute of Physics and Technology,
Institutskii Per. 9, Dolgoprudny, Moscow Region 141700, Russia}

\affiliation{Institute of Physics and Technology, Russian Academy
of Sciences, Nakhimovskii Pr. 34, Moscow 117218, Russia}

\affiliation{Russian Quantum Center, Novaya 100, Skolkovo, Moscow
Region 143025, Russia}

\affiliation{P.~N.~Lebedev Physical Institute, Russian Academy of
Sciences, Leninskii Pr. 53, Moscow 119991, Russia}

\begin{abstract}
Entanglement degradation in open quantum systems is reviewed in
the Choi-Jamio{\l}kowski representation of linear maps. In
addition to physical processes of entanglement dissociation and
entanglement annihilation, we consider quantum dynamics
transforming arbitrary input states into those that remain
positive under partial transpose (PPT-inducing channels). Such
evolutions form a convex subset of distillation-prohibiting
channels. A relation between the above channels and
entanglement-binding ones is clarified. An example of the
distillation-prohibiting map $\Phi\otimes\Phi$ is given, where
$\Phi$ is not entanglement binding.
\end{abstract}


\maketitle

\section{\label{introduction} Introduction}
 The phenomenon of quantum entanglement usually emerges
between interacting subsystems in a composite system. The
long-distance and long-living forms of such correlations play a
vital role in up-to-date quantum
technologies~\cite{horodecki-2009}. Unavoidable interaction with
the environment changes the structure of
entanglement~\cite{aolita-2014}. In special cases, the global
environment bath can create quantum correlations between the
particles of the composite system~\cite{palma-1989}, however,
local noises degrade it and impose limitations on achievable
entanglement death time in
experiments~\cite{almeida-2007,bellomo-2007}. Bound entangled
states are those that are still entangled but cannot be distilled
into maximally entangled qubit pairs~\cite{horodecki-1998}.
Positive partial transpose (PPT) states~\cite{peres-1996} are
known to be undistillable, so it is reasonable to characterize
quantum channels resulting in PPT states regardless of the input.

To describe the stages of gradual entanglement degradation, the
notions of entanglement
annihilation~\cite{moravcikova-ziman-2010,filippov-rybar-ziman-2012}
and entanglement dissociation~\cite{filippov-melnikov-ziman-2013}
were introduced recently. PPT-inducing and
distillation-prohibiting channels can be considered as their
generalizations to corresponding entanglement properties (PPT and
undistillability). The previously introduced notion of
entanglement-binding channel~\cite{horodecki-2000} turns out to be
a partial case (one-sided realization) of the
distillation-prohibiting channel.

The paper is organized as follows.

In Sec.~\ref{section:2}, the description of quantum entanglement
is briefly reviewed. In Sec.~\ref{section:3}, the basic
information about quantum channels and their entanglement
degradation properties is given. In Sec.~\ref{section:4}, the
Choi-Jamio{\l}kowski
representation~\cite{pillis-1967,jamiolkowski-1972,choi-1975} of
linear maps and their concatenations is discussed, the structure
of Choi matrix is reviewed for entanglement-breaking,
entanglement-binding, entanglement-annihilating, and
entanglement-dissociating channels. Main results are presented in
Sec.~\ref{section:5}, where properties of PPT-inducing and
distillation-prohibiting channels are studied. In
Sec.~\ref{section:6}, brief conclusions are given.

\section{\label{section:2} Quantum entanglement}
 A quantum state is described by a density operator $\varrho$,
which is Hermitian, positive semidefinite, and has unit trace. In
a composite system, there are several degrees of freedom, say,
$A,B,C,\ldots$, and the corresponding Hilbert space has the tensor
product structure
$\mathcal{H}^A\otimes\mathcal{H}^B\otimes\mathcal{H}^C\otimes\cdots$.
Specification of particular degrees of freedom fixes the
partitioning structure $A|B|C|\cdots$ and the associated
entanglement structure \footnote{For example, a two-body system
can be described by either coordinates of individual particles
${\bf r}_1$ and ${\bf r}_2$, or the center-of-mass coordinate
${\bf R} = \frac{m_1{\bf r}_1 + m_2{\bf r}_2}{m_1+m_2}$ and the
relative coordinate ${\bf r} = {\bf r}_2 - {\bf r}_1$. The ground
state of the hydrogen atom is entangled with respect to the first
partition and separable with respect to the second
one~\cite{dugic}}. In experiments and applications, one deals with
the accessible degrees of freedom that are naturally related to
the used measurement techniques. From now on, these degrees of
freedom are supposed to be fixed.

The state $\varrho^{ABC\cdots}$ is called fully separable if there
exist a probability distribution $\{p_k\}$ and density operators
$\varrho_k^A, \varrho_k^B, \varrho_k^C, \ldots$ such that
$\varrho^{ABC\cdots}$ belongs to the closure of the states $\sum_k
p_k \varrho_k^A \otimes \varrho_k^B \otimes \varrho_k^C \otimes
\cdots$; otherwise, $\varrho^{ABC\cdots}$ is called
entangled~\cite{werner-1989}.

An example of a coarse-grained partition is
$\mathcal{P}=AB|C|DEF$. The state $\varrho^{ABCDEF}$ is separable
with respect to $\mathcal{P}$ if $\varrho^{ABCDEF} = \sum_k p_k
\varrho_k^{AB} \otimes \varrho_k^C \otimes \varrho_k^{DEF}$. A
fully separable state is separable with respect to $AB|C|DEF$ but
the converse does not hold in general. It is instructive to remind
about the existence of a three-qubit state $\varrho^{ABC}$, which
is separable with respect to all bipartitions ($A|BC$, $B|AC$, and
$C|AB$) but is not fully separable, i.e. entangled with respect to
tripartition $A|B|C$~\cite{bennet-1999,divincenzo-2003}.

A state, which cannot be written as a convex sum of any separable
bipartite states, is usually referred to as genuinely entangled.
Under the action of local noises, the genuine entanglement
degrades at first to the biseparable form (convex sum of states
separable with respect to bipartitions), then to the triseparable
form (convex sum of states separable with respect to
tripartitions), and so on. Such a process is called entanglement
dissociation~\cite{filippov-melnikov-ziman-2013} in analogy to the
dissociation of chemical compounds in solvents. The final stage of
entanglement evolution is annihilation, when the system state
becomes fully separable.

A bipartite state $\varrho^{AB}$ is called positive under partial
transpose (PPT) if $\varrho^{A(B)^{\top}} \equiv {\rm Id}^A
\otimes T^B[\varrho^{AB}]$ is positive
semidefinite~\cite{peres-1996}. Here, ${\rm Id}$ is the identity
transformation and $T^B$ is the transposition in some orthonormal
basis in $\mathcal{H}^B$. Clearly, $\varrho^{A(B)^{\top}} \ge 0
\Longleftrightarrow \varrho^{(A)^{\top}B} \ge 0$ because
$\varrho^{(AB)^{\top}}$ is a valid density operator and $T^2 =
{\rm Id}$. A separable state $\varrho^{AB}$ is necessarily
PPT~\cite{peres-1996}, the converse holds if $d^A d^B \le
6$~\cite{horodecki-1996}. PPT entangled states are known to be
undistillable~\cite{horodecki-1998,horodecki-2001}, which
motivates us to study PPT-inducing quantum evolutions.

\section{\label{section:3} Quantum channels}

 Evolution of an open quantum system during the time interval
$(0,T)$ can be considered as an input-output relation
$\varrho_{t=T} = \Phi [\varrho_{t=0}]$ between the final and
initial density matrices. If the system is decoupled from the
environment at time $t=0$, then $\Phi$ is a single-valued linear
map. Physical reasoning leads to a conclusion that $\Phi$ is a
completely positive trace preserving
map~\cite{stinespring,sudarshan-1961}, which we will refer to as
quantum channel (see, e.g.,~\cite{breuer-petruccione}).

Suppose a channel $\Phi$ acting on the system $S=ABC\ldots$ and
the initial state $\varrho_{\rm in}$. If the output state
$\varrho_{\rm out} = \Phi[\varrho_{\rm in}]$ is separable with
respect to a fixed partition $\mathcal{P}$ for all $\varrho_{\rm
in}$, then $\Phi$ is called entanglement-dissociating with respect
to $\mathcal{P}$. If $\mathcal{P}=A|B|C|\cdots$, i.e.
$\varrho_{\rm out}$ is fully separable for all input states, then
$\Phi$ is called
entanglement-annihilating~\cite{filippov-melnikov-ziman-2013}.

Notions of entanglement dissociation and annihilation do not imply
the use of any auxiliary systems. In contrast, the so called
entanglement breaking channels are those for which
$\Phi\otimes{\rm Id}[\varrho_{\rm in}^{S+{\rm anc}}]$ is separable
with respect to partition $S|{\rm anc}$ for any initial density
operator $\varrho_{\rm in}^{S+{\rm anc}}$, with the dimension of
an ancillary Hilbert space being
arbitrary~\cite{holevo-1998,horodecki-2003,holevo-2008}.
Equivalently, $\Phi^S$ is entanglement-breaking if
$\Phi^S\otimes{\rm Id}^{\rm anc}$ is entanglement dissociating
with respect to the bipartition $S|{\rm anc}$ for all dimensions
${\rm dim}\mathcal{H}^{\rm anc} =2,3,\ldots$.

Similarly, if $\Phi\otimes{\rm Id}[\varrho_{\rm in}^{S+{\rm
anc}}]$ is undistillable (with respect to the system $S$ and the
ancillary system) for any initial density operator $\varrho_{\rm
in}^{S+{\rm anc}}$, then $\Phi$ is called entanglement
binding~\cite{horodecki-2000}.

\section{\label{section:4} Choi-Jamio{\l}kowski representation}
 In case of finite dimensions, a linear map $\Phi$ acting on a
system $S$ can be defined via the so-called Choi-Jamio{\l}kowski
isomorphism~\cite{pillis-1967,jamiolkowski-1972,choi-1975}:
\begin{eqnarray}
&& \label{choi-matrix} \Omega_{\Phi}^{SS'} = \Phi^S \otimes {\rm
Id}^{S'} [\ket{\Psi_+^{SS'}}\bra{\Psi_+^{SS'}}], \\
&& \label{map-through-choi} \Phi [X] = d^S \, {\rm tr}_{S'} \, [
\,\Omega_{\Phi}^{SS'} (I_{\rm out}^S \otimes X^{\rm T}) \, ],
\end{eqnarray}

\noindent where $d^S={\rm dim} \mathcal{H}^S = {\rm dim}
\mathcal{H}^{S'}$, $\ket{\Psi_+^{SS'}} = (d^S)^{-1/2}
\sum_{i=1}^{d^S} \ket{i\otimes i'}$ is a maximally entangled state
shared by system $S$ and its clone $S'$, ${\rm tr}_{S'}$ denotes
the partial trace over $S'$, $I$ is the identity operator, and
$X^{\rm T}=\sum_{i,j} \bra{j} X \ket{i} \ket{i'}\bra{j'}$ is the
transposition in some orthonormal basis. The
operator~\eqref{choi-matrix} is referred to as Choi operator or
Choi state (see Fig.~\ref{figure}).

\begin{figure*}
\includegraphics[width=17cm]{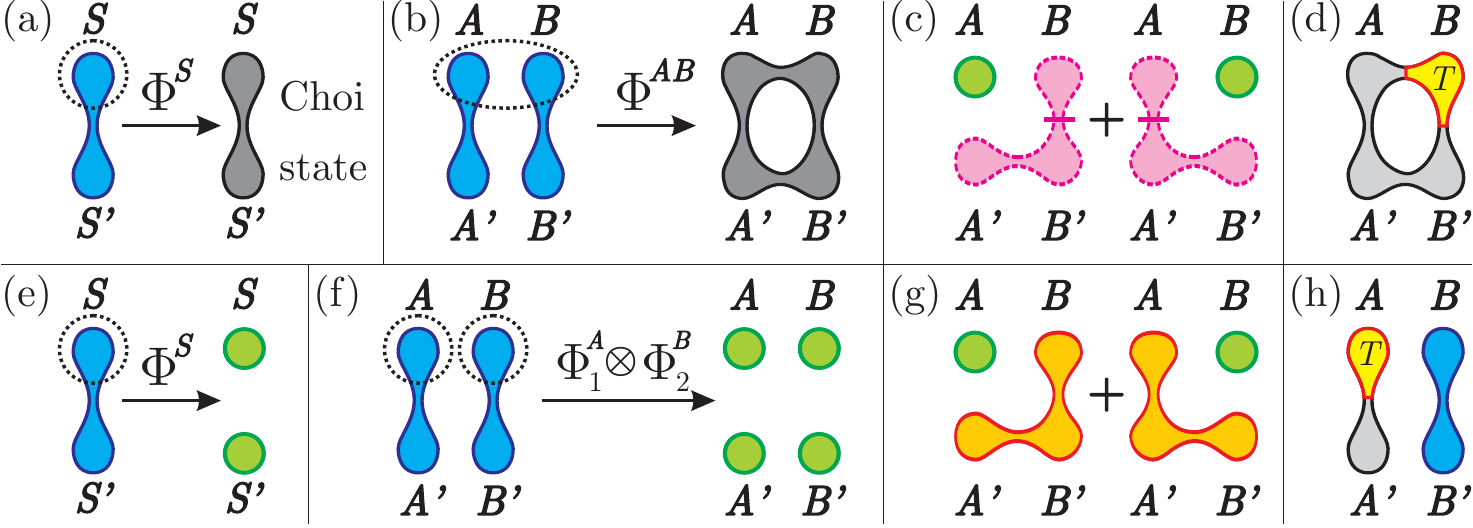}
\caption{\label{figure} Choi operator $\Omega_{\Phi}^{SS'}$ as a
result of applying $\Phi^S$ to one side of the state
$\ket{\Psi_{+}^{SS'}}$ maximally entangled between system $S$ and
its clone $S'$ (a). Choi operator for a map $\Phi^{AB}$ acting on
a composite system $AB$ (b). Visualization of property 5 for Choi
matrices of entanglement-annihilating channels (c). Visualization
of Corollary~\ref{corollary-1} for PPT-inducing channels (d). Choi
state of entanglement-breaking channel is separable, Choi state of
entanglement binding channel is undistillable (e). Choi state of
local entanglement-breaking channel $\Phi_1 \otimes \Phi_2$ is
fully separable (f). Visualization of property 6 --- sufficient
condition of entanglement annihilation --- in terms of Choi
matrices (g). Visualization of Proposition~\ref{proposition-2}
depicting the Choi matrix of PPT-inducing channel
$\Phi^A\otimes{\rm Id}^B$ (h).}
\end{figure*}

\subsection{Review of properties}

The well known results are as follows:
\begin{enumerate}
\item $\Phi$ is positive if and only if $\Omega_{\Phi}^{SS'}$ is
block-positive, i.e. $\bra{\varphi^S \otimes \chi^{S'}}
\Omega_{\Phi}^{SS'} \ket{\varphi^S \otimes \chi^{S'}} \ge 0$ for
all $\varphi,\chi$~\cite{jamiolkowski-1972}.

\item $\Phi$ is completely positive (quantum operation) if and
only if $\Omega_{\Phi}^{SS'} \ge 0$~\cite{choi-1975}.

\item $\Phi$ is entanglement breaking if and only if
$\Omega_{\Phi}^{SS'} \ge 0$ is separable with respect to partition
$S|S'$~\cite{horodecki-2003}.

\item $\Phi$ is entanglement binding if and only if
$\Omega_{\Phi}^{SS'} \ge 0$ is undistillable (with respect to
parties $S$ and $S'$)~\cite{horodecki-2000}.
\end{enumerate}

A feature of entanglement breaking maps is that their outcome is
separable not only for positive inputs (density operators)
$\varrho^{S+{\rm anc}}$ but also for block-positive inputs
$\xi_{\rm BP}^{S|{\rm anc}}$. Since ${\rm tr}[\Omega_{\rm pos}
\Omega_{\rm ent-br}] \ge 0$, the cone of entanglement breaking
maps is dual to the cone of positive maps (see,
e.g.,~\cite{bengtsson-zyczkowski,johnston-2012}).

As far as multipartite composite systems $S=ABC\ldots$ are
concerned, the maximally entangled state can be viewed as
separable with respect to the partition $AA'|BB'|CC'|\ldots$,
namely, $\ket{\Psi_+^{SS'}} = (d^A d^B d^C \cdots )^{-1/2}
\sum_{i=1}^{d^A} \sum_{j=1}^{d^B} \sum_{k=1}^{d^C} \sum_{\cdots}
\ket{ijk\cdots}\otimes \ket{i'j'k'\cdots} = \ket{\Psi_+^{AA'}}
\otimes \ket{\Psi_+^{BB'}} \otimes \ket{\Psi_+^{CC'}} \otimes
\cdots$.

If $\Phi$ is a local map, i.e. has the form $\Phi_1^A \otimes
\Phi_2^B \otimes \Phi_3^C \otimes \cdots$, then the composite Choi
operator reads $\Omega_{\Phi_1 \otimes \Phi_2 \otimes \Phi_3
\otimes \ldots}^{ABC \ldots A'B'C' \ldots} = \Omega_{\Phi_1}^{AA'}
\otimes \Omega_{\Phi_2}^{BB'} \otimes \Omega_{\Phi_3}^{CC'}
\otimes \cdots$. Positivity of this operator is equivalent to
positivity of individual Choi operators. By property 2, $\Phi_1^A
\otimes \Phi_2^B \otimes \Phi_3^C \otimes \cdots$ is completely
positive if and only if each of the maps $\Phi_1^A$, $\Phi_2^B$,
$\Phi_3^C$, $\ldots$ is completely positive. Similarly, by
properties 3 and 4, $\Phi_1^A \otimes \Phi_2^B \otimes \Phi_3^C
\otimes \cdots$ is entanglement breaking (binding) if and only if
each of the maps $\Phi_1^A$, $\Phi_2^B$, $\Phi_3^C$, $\ldots$ is
entanglement breaking (binding). Nonetheless, property 1 cannot be
extended in analogous way. In fact, the map $\Phi_1^A \otimes
\Phi_2^B$ is positive if and only if $\Omega_{\Phi_1}^{AA'}
\otimes \Omega_{\Phi_2}^{BB'} = \xi_{\rm BP}^{AB|A'B'}$, i.e. the
composite Choi matrix is block positive with respect to partition
$AB|A'B'$.

Suppose a positive map $\Upsilon$ acting on a composite system
$ABC\ldots$. $\Upsilon$ dissociates entanglement with respect to
some partition $\mathcal{P}(ABC\ldots)$ if and only if
\begin{equation}
\label{PEA} {\rm tr} \left[ \Omega_{\Upsilon}^{ABC\ldots
A'B'C'\ldots} \left( \xi_{\rm BP}^{\mathcal{P}(ABC\ldots)} \otimes
\varrho^{A'B'C'\ldots} \right) \right] \ge 0
\end{equation}

\noindent for all block-positive $\xi_{\rm
BP}^{\mathcal{P}(ABC\ldots)}$ and density operators
$\varrho^{A'B'C'\ldots}$. This result follows immediately from
Eq.~\eqref{map-through-choi}~\cite{filippov-melnikov-ziman-2013}.
However, in order to describe a valid quantum evolution,
$\Upsilon$ has to be completely positive, i.e.
$\Omega_{\Upsilon}^{ABC\ldots A'B'C'\ldots} \ge 0$. Choi matrices
that satisfy this condition and requirement~\eqref{PEA} are
precisely entanglement dissociating channels with respect to the
partition $\mathcal{P}$.

The bipartite setting of Eq.~\eqref{PEA} reads ${\rm tr} [
\Omega_{\Phi}^{ABA'B'} (\xi_{\rm BP}^{A|B} \otimes \varrho^{A'B'})
] \ge 0$ and specifies entanglement annihilating maps $\Phi^{AB}$.
The last formula shows that a cone of entanglement annihilating
maps is dual to the cone of maps $\Theta^{AB}$ of the form
$\Theta[X] = \sum_{k} \tr{F_k X} \xi_{{\rm BP}~k}^{A|B}$, $F_k \ge
0$.

A map $\Phi^{\dag}$ is called dual to the map $\Phi$ if ${\rm
tr}\left[ \Phi^{\dag}[X] Y \right] \equiv {\rm tr}\left[ X \Phi[Y]
\right]$ for $X,Y$ from corresponding domains\footnote{Mind the
difference between the notion of dual cones of maps and the notion
of dual maps.}. An observation for the entanglement annihilating
map $\Phi^{AB}$ is that $\Phi^{\dag}$ transforms all
block-positive operators $\xi_{\rm BP}^{A|B}$ into positive ones
($\propto \varrho^{AB}$). As a consequence, the concatenation
$\Phi \circ \Phi^{\dag}$ has to map block-positive operators to
separable ones (this is a necessary condition for $\Phi^{AB}$ to
be entanglement annihilating).

Finally, sufficient conditions for entanglement annihilation are
as follows~\cite{filippov-ziman-2013}:
\begin{enumerate}
\item[5.] If $\Omega_{\Phi}^{ABA'B'}$ can be represented in the
form of a convex sum of operators $\xi_{\rm BP}^{A|A'B'} \otimes
\varrho^{B}$ and $\varrho^{A} \otimes \xi_{\rm BP}^{B|A'B'}$, then
$\Phi^{AB}$ annihilates entanglement.

\item [6.]  If $\Omega_{\Phi}^{ABA'B'}$ is a convex sum of
separable states of the form $\varrho^{A|BA'B'}$ and
$\varrho^{B|AA'B'}$, then $\Phi^{AB}$ is entanglement
annihilating.
\end{enumerate}

Some key properties of Choi operators for the discussed maps are
depicted schematically in Fig.~\ref{figure}.

\subsection{Concatenation of maps in terms of Choi operators}
 A concatenation $\Phi \circ \Xi$ of two maps $\Phi:
\mathcal{M}_d \mapsto \mathcal{M}_d$ and $\Xi: \mathcal{M}_d
\mapsto \mathcal{M}_d$ is a map such that $\Phi \circ \Xi [X]
\equiv \Phi \left[\Xi [X]\right]$. It is not hard to see that
\begin{equation}
\label{star} \Omega_{\Phi \circ \Xi} = d \!\!\!\!\!\!
\sum_{i,j,k,l,m,n=1}^d \!\!\!\!\!\! \ket{m \otimes k} \bra{m
\otimes i} \Omega_{\Phi} \ket{n \otimes j} \bra{i \otimes k}
\Omega_{\Xi} \ket{j \otimes l} \bra{n \otimes l}.
\end{equation}

\noindent The rule~\eqref{star} can be treated as a star-product
scheme~\cite{manko-2002,manko-2007}, where the symbols are
elements of Choi matrices. The kernel of the star product reads
$K(mk,nl;pq,rs;tu,vw)=\delta_{mp}\delta_{qt}\delta_{rn}\delta_{sv}\delta_{ku}\delta_{wl}$.

\section{\label{section:5} PPT-inducing and distillation-prohibiting channels}
 A map $\Phi^{AB}$ transforming operators acting on
$\mathcal{H}^A\otimes\mathcal{H}^B$ is called PPT-inducing if
$\Phi^{AB}[\varrho^{AB}]$ is positive and PPT with respect to the
partition $A|B$ for all input states $\varrho^{AB}$. If in
addition $\Phi^{AB}$ is completely positive and trace preserving,
then $\Phi^{AB}$ is the PPT-inducing quantum channel.

\begin{proposition}
\label{proposition-1} A positive map $\Phi^{AB}$ is PPT-inducing
if and only if $\Omega_{\Phi}^{A(B)^{\top}A'B'}$ is block-positive
with respect to the partition $AB|A'B'$.
\end{proposition}
\begin{proof}
Partial transposition of formula~\eqref{map-through-choi} yields
$\varrho_{\rm out}^{A(B)^{\top}} = {\rm tr}_{A'B'}\left[
\Omega_{\Phi}^{A(B)^{\top}A'B'} (I^{AB} \otimes \varrho_{\rm
in}^{(AB)^{\top}}) \right]$. Positivity of $\varrho_{\rm
out}^{A(B)^{\top}}$ means ${\rm tr}[\varrho_{\rm
out}^{A(B)^{\top}} \widetilde{\varrho}^{AB}] \ge 0$ for all
density operators $\widetilde{\varrho}^{AB}$, which is equivalent
to ${\rm tr} \left[ \Omega_{\Phi}^{A(B)^{\top}A'B'}
(\widetilde{\varrho}^{AB} \otimes \varrho_{\rm in}^{(AB)^{\top}})
\right] \ge 0$ for any $\widetilde{\varrho}^{AB}$ and
$\varrho_{\rm in}^{AB}$, i.e. block-positivity of
$\Omega_{\Phi}^{A(B)^{\top}A'B'}$ with respect to the partition
$AB|A'B'$.
\end{proof}

Apparently, one could use $\Omega_{\Phi}^{(A)^{\top}BA'B'}$
instead of $\Omega_{\Phi}^{A(B)^{\top}A'B'}$ in the formulation of
Proposition~\ref{proposition-1}. Besides, the positivity of map
$\Phi^{AB}$ is essential.

Since positive operators are automatically block-positive, we
readily obtain the following result.

\begin{corollary}
\label{corollary-1} Suppose a channel $\Phi^{AB}$ such that
$\Omega_{\Phi}^{A(B)^{\top}A'B'} \ge 0$, then $\Phi^{AB}$ is
PPT-inducing.
\end{corollary}

Another sufficient condition can be found by using the map
decomposition technique developed in~\cite{filippov-ziman-2013}.
In fact, {\it $\Phi^{AB}$ is PPT-inducing if it can be decomposed
into $\Phi^{AB} = \sum_k (\mathcal{O}_k^A \otimes {\rm Id}^B)
\circ \Lambda_k^{AB}$, where the maps $\Lambda_k^{AB}$ are
positive and the operations $\mathcal{O}_k^A \otimes {\rm Id}^B$
are PPT-inducing.} This fact can be also proven by the
rule~\eqref{star} and Proposition~\ref{proposition-1}.

\begin{proposition}
\label{proposition-2} A channel $\Phi^{A}\otimes {\rm Id}^B$ with
${\rm dim}\mathcal{H}^B  \ge {\rm dim}\mathcal{H}^A$ is
PPT-inducing if and only if $\Omega_{\Phi}^{AA'}$ is PPT.
\end{proposition}
\begin{proof}
Necessity follows from the fact that the maximally entangled state
$\ket{\Psi_{+}^{AB}}$ should become PPT when acted upon by
$\Phi^{A}\otimes {\rm Id}^B$. Identification of the proper
$d^A$-dimensional subspace of $\mathcal{H}^B$ and
$\mathcal{H}^{A'}$ implies $\Omega_{\Phi}^{AA'}$ is PPT.
Sufficiency follows from the tensor product form of the Choi
operator, namely, $\Omega_{\Phi\otimes{\rm
Id}}^{ABA'B'}=\Omega_{\Phi}^{AA'}\otimes\ket{\Psi_{+}^{BB'}}\bra{\Psi_{+}^{BB'}}$,
which implies $\Omega_{\Phi\otimes{\rm
Id}}^{(A)^{\top}BA'B'}=\Omega_{\Phi}^{(A)^{\top}A'}\otimes\ket{\Psi_{+}^{BB'}}\bra{\Psi_{+}^{BB'}}
\ge 0$ because $\Omega_{\Phi}^{AA'}$ is PPT. Thus,
$\Omega_{\Phi\otimes{\rm Id}}^{(A)^{\top}BA'B'}$ is positive,
which guarantees the PPT-inducing behavior of the channel
$\Phi^{AB}$ by Corollary~\ref{corollary-1}.
\end{proof}

To some extent, Proposition~\ref{proposition-2} reproduces the
result of Ref.~\cite{horodecki-2000} and serves as a sufficient
condition for the map $\Phi$ to be entanglement binding. Examples
of PPT but entangled Choi states $\Omega_{\Phi}^{AA'}$ are given
in Ref.~\cite{horodecki-2000}.

A map $\Phi^{AB}$ that maps density operators to undistillable
ones is called distillation-prohibiting. As PPT states cannot be
distilled~\cite{horodecki-1998}, PPT-inducing maps form a convex
subset of distillation-prohibiting ones.

The following example illustrates the relation between local
entanglement binding channels and distillation-prohibiting ones.

\begin{example}
Consider a local channel $\Phi_q \otimes \Phi_q$ acting on two
qutrits, where $\Phi_q$ is depolarizing: $\Phi_q[X] = q X + (1-q)
{\rm tr}[X] \frac{I}{3}$, $q \ge - \frac{1}{8}$. Suppose
$\ket{\psi} = \sum_{i=1}^3 \sqrt{\lambda_i}
\ket{\varphi_i}\ket{\chi_i}$ is the Schmidt decomposition of an
input pure state $\ket{\psi}$, i.e. $\ip{\varphi_i}{\varphi_j} =
\ip{\chi_i}{\chi_j} = \delta_{ij}$, $0 \le \lambda_i \le 1$, and
$\sum_{i=1}^{3} \lambda_i = 1$. Written in the basis
$\{\ket{\varphi_i}\ket{\chi_j}\}$, the density operator $\Phi_q
\otimes \Phi_q[\ket{\psi}\bra{\psi}]$ is a sparse matrix, so one
can readily find eigenvalues of its partial transpose. For a fixed
$q$, the minimal possible eigenvalue is achieved if
$\lambda_1=\lambda_2=\frac{1}{2}$ and $\lambda_3=0$ (up to the
permutation of indexes). Exploring positivity of the minimal
eigenvalue, we find that $\Phi_q \otimes \Phi_q$ is PPT-inducing
if and only if $q \le \frac{1+\sqrt{3}}{4+\sqrt{3}} \approx
0.4766$. However, $\Phi_q$ is known to be entanglement binding if
and only if $q \le \frac{1}{4} =
0.25$~\cite{horodecki-isotropic-1999}. A gap between these two
values shows that a local channel $\Phi_1 \otimes \Phi_2$ can be
PPT-inducing (distillation-prohibiting) even if neither of
channels $\Phi_1$ or $\Phi_2$ is entanglement binding. \hfill
$\blacksquare$
\end{example}

Assuming that the state
$\ket{\psi}=\frac{1}{\sqrt{2}}(\ket{\varphi}\ket{\chi} +
\ket{\varphi_{\perp}}\ket{\chi_{\perp}})$ of Schmidt rank 2
minimizes the eigenvalues of the partial transpose of $\Phi_q
\otimes \Phi_q[\ket{\psi}\bra{\psi}]$ for a general depolarizing
map $\Phi_q: \mathcal{M}_d \mapsto \mathcal{M}_d$, $\Phi_q[X] = q
X + (1-q) {\rm tr}[X] \frac{I}{d}$, $q \ge -\frac{1}{d^2-1}$, we
readily obtain the following result.

{\bf Conjecture}. $\Phi_q \otimes \Phi_q$ is PPT-inducing if $q
\le \frac{1+\sqrt{3}}{d+1+\sqrt{3}}$.

On the other hand, $\Phi_q$ is entanglement binding if and only if
$q \le \frac{1}{d+1}$~\cite{horodecki-isotropic-1999}. The bound
on parameter $q$ in the above Conjecture was first found in
Ref.~\cite{filippov-ziman-2013} in connection with the search of
robust entangled states.

\section{\label{section:6} Conclusions}
 We have introduced the concepts of PPT-inducing and
distillation-prohibiting channels as those that act on a composite
system $AB$ and result in PPT and undistillable output states,
respectively. When a one-sided noisy process $\Phi^A \otimes {\rm
Id}^B$ is distillation-prohibiting, then our definition naturally
leads to the notion of entanglement binding channel $\Phi^A$. We
have characterized PPT-inducing channels and found necessary and
sufficient conditions for $\Phi^{AB}$ to be PPT-inducing in terms
of Choi operators; however, distillation-prohibiting channels
still need further analysis.

Also, we have demonstrated that $\Phi_1 \otimes \Phi_2$ can be
PPT-inducing (distillation-prohibiting) even if neither of
channels $\Phi_1$ or $\Phi_2$ is entanglement binding. All these
results show that entanglement binding channels are analogues of
entanglement breaking ones, whereas distillation-prohibiting
channels are analogues of entanglement annihilating ones: in both
cases the notion of entanglement is merely replaced by the notion
of distillation capability.

Recent progress in the description of the sets of PPT and
undistillable states~\cite{brandao-eisert-2008,augusiak-2010} may
stimulate a further characterization of PPT-inducing and
distillation-prohibiting channels. Future investigation of
continuous-variable systems may follow the same line of reasoning
as in Refs.~\cite{filippov-ziman-2013,filippov-ziman-2014} because
a necessary and sufficient condition for PPT continuous-variable
states is known~\cite{shchukin-2005}.

\begin{acknowledgments}
The author is grateful to M\'{a}rio Ziman for encouraging
comments. The study is partially supported by the Russian
Foundation for Basic Research under project No. 14-07-00937-a.
\end{acknowledgments}

\end{document}